\documentclass[letterpaper, 10 pt, conference]{ieeeconf}  % Comment this line out if you need a4paper

\IEEEoverridecommandlockouts                              % This command is only needed if 
\usepackage[utf8]{inputenc}
\usepackage[english]{babel}
\usepackage[ruled,linesnumbered]{algorithm2e}
\usepackage{amsmath}
\usepackage{amssymb}
\usepackage{array}
\usepackage[backend=bibtex,sorting=none]{biblatex}
\usepackage{booktabs}
\usepackage{caption}
\usepackage{dblfloatfix}
\usepackage{enumerate}
\usepackage{gensymb}
\usepackage{graphicx}
\usepackage{multirow}
\usepackage{indentfirst}
\usepackage{tabularx}
\usepackage{threeparttable}
\usepackage{subcaption}
\usepackage{mathtools}

\usepackage{color}
\usepackage{hyperref}
\usepackage{kotex}

\usepackage{amsthm}
\newtheorem{theorem}{Theorem}

\title{\LARGE \bf
Efficient Multi-Agent Trajectory Planning with Feasibility Guarantee using Relative Bernstein Polynomial
}

\author{Jungwon Park, Junha Kim, Inkyu Jang and H. Jin Kim$^{1}$% <-this % stops a space
\thanks{This work was supported by Institute of Information \& Communications Technology Planning \& Evaluation(IITP) grant funded by the Korea government(MSIT) (No. 2019-0-00399, Development of A.I. based recognition, judgement and control solution for autonomous vehicle corresponding to atypical driving environment)}
\thanks{$^{1}$The authors are with the Department of Mechanical and Aerospace Engineering, Seoul National University, Seoul, South Korea.
        {\tt\small \{qwerty35, wnsgk02, leplusbon, hjinkim\}@snu.ac.kr}}%
}

\begin{document}

\maketitle
\thispagestyle{empty}
\pagestyle{empty}

%%%%%%%%%%%%%%%%%%%%%%%%%%%%%%%%%%%%%%%%%%%%%%%%%%%%%%%%%%%%%%%%%%%%%%%%%%%%%%%%
\begin{abstract}
This paper presents a new efficient algorithm which guarantees a solution for a class of multi-agent trajectory planning problems in obstacle-dense environments.
Our algorithm combines the advantages of both grid-based and optimization-based approaches, and generates safe, dynamically feasible trajectories without suffering from an erroneous optimization setup such as imposing infeasible collision constraints.
We adopt a sequential optimization method with \textit{dummy agents} to improve the scalability of the algorithm, and utilize the convex hull property of Bernstein and relative Bernstein polynomial to replace non-convex collision avoidance constraints to convex ones.
The proposed method can compute the trajectory for 64 agents on average 6.36 seconds with Intel Core i7-7700 @ 3.60GHz CPU and 16G RAM, 
and it reduces more than $50\%$ of the objective cost compared to our previous work.
We validate the proposed algorithm through simulation and flight tests.
\end{abstract}

%%%%%%%%%%%%%%%%%%%%%%%%%%%%%%%%%%%%%%%%%%%%%%%%%%%%%%%%%%%%%%%%%%%%%%%%%%%%%%%%
\section{INTRODUCTION} 
Multi-agent systems with many unmanned aerial vehicles (UAVs) broaden the range of achievable missions to complex environments unsafe or hard to reach for humans or a single agent. 
For successful operation of these multi-agent systems, path planning algorithm is required to generate a collision-free trajectory in any obstacle environment. However, many works have a risk to fail in dense cluttered environments due to deadlock \cite{bareiss2013reciprocal, zhou2017fast} or failure caused by enforcing infeasible collision constraints in the formulation \cite{chen2015decoupled, luis2019trajectory}.

In this paper, we present an efficient multi-agent trajectory planning algorithm which generates safe, dynamically feasible trajectories in obstacle-dense environments by extending our previous work \cite{1909.02896}.
The proposed algorithm is designed to have the advantages of both grid-based and optimization-based approaches. First, it guarantees the feasibility of optimization problem formulation by utilizing an initial trajectory computed from grid-based multi-agent path finding algorithm. Second, it generates a continuous dynamically feasible trajectory by optimizing the initial trajectory with consideration of quadrotor dynamics as shown in Fig. \ref{fig: flight test}.
When we formulate the optimization problem, we utilize the convex hull property of relative Bernstein polynomial to translate non-convex collision avoidance constraints to convex ones.
Compared to the previous work \cite{1909.02896}, we modify the method for constructing constraints not to occur infeasible constraints between collision avoidance constraints, and we introduce a sequential optimization method. This sequential method can deal with a large scale of agents with improved computational efficiency, and does not cause deadlock by employing \textit{dummy agents}.

Our main contributions can be summarized as follows.
\begin{itemize}
\item A multi-agent trajectory planning algorithm is presented for obstacle-dense environments, which generates collision-free and dynamically feasible trajectories without a potential optimization failure by using relative Bernstein polynomial.
\item A sequential trajectory optimization method is proposed with dummy agents, which reduces computational load.
\item The source code will be released in \url{https://github.com/qwerty35/swarm_simulator}.
\end{itemize}

\begin{figure}[t]
\centering
\includegraphics[width = 0.60\linewidth]{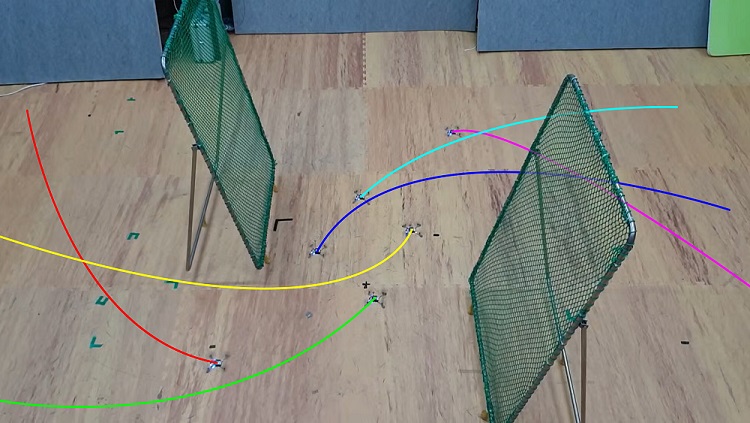}
\caption{
Flight in an obstacle environment with 6 quadrotors.
}
\label{fig: flight test}
\vspace{-5mm}
\end{figure}

% \section{RELATED WORK}
There have been discussions in literature  closely related to our work on multi-agent trajectory planning. 
In \cite{mellinger2012mixed, augugliaro2012generation, chen2015decoupled}, the trajectory generation problems are reformulated as mixed-integer quadratic programming (MIQP) or sequential convex programming (SCP) problems, that apply collision constraints at each discrete time step.
These methods suit well systems with a small number of agents, but they are intractable for large teams and complex environments because an additional adaptation process is required to find proper discretization time step depending on the size of agents and obstacles. On the other hand our method does not require this process because we do not use time discretization.

Sequential planning proposed in \cite{robinson2018efficient} for better scalability  is similar to our work. However, it may not be able to find a feasible solution for a crowded situation. To solve this, we adopt dummy agents which move along the initial trajectory computed by a grid-based planner to prevent deadlock. %(?? 실제 초기 궤적인지 optimization 초기 궤적인지 불분명)
The most relevant work can be found in \cite{honig2018trajectory, debord2018trajectory}. They plan an initial trajectory with a grid-based planner and then construct a safe flight corridor (SFC), which indicates a safe region of each agent. However, they need to resize SFC iteratively until the overall cost converges, while our proposed method does not need an additional resizing process by using relative Bernstein polynomial.

Recently, distributed planning is receiving much attention due to scalability \cite{bareiss2013reciprocal, zhou2017fast, luis2019trajectory}.
However, such distributed methods are not able to guarantee a safe solution in obstacle-dense environments due to deadlock.

% \textcolor{green}{
% This paper is structured as follows.
% The problem formulation is presented in section \ref{sec: problem formulation}. 
% In section \ref{sec: relative bernstein polynomial}, we introduce relative Bernstein polynomial to deal with inter-collision avoidance. 
% We define key terms used in this paper in section \ref{sec: definition}, and we describe the method of multi-agent trajectory planning in section \ref{sec: method}.
% Experimental results are presented in section \ref{sec: experiments}.
% Finally, section \ref{sec: conclusions} contains conclusions.
% }

\section{PROBLEM FORMULATION}
\label{sec: problem formulation}
In this section, we formulate an optimization problem to generate safe, continuous trajectories for a multi-agent robot system consisting of $N$ quadrotors. We assign the mission for the $i^{th}$ quadrotor to move from the start point $s^i$ to the goal point $g^i$. The quadrotors may have a different size with radius $r^{1},...,r^{N}$m. The maximum velocity and acceleration of the $i^{th}$ quadrotor are $v^{i}_{max}$, $a^{i}_{max}$ respectively.

\subsection{Assumption}
\label{subsec: assumption}
We assume that prior knowledge of the free space $\mathcal{F}$ of the environment is given as a 3D occupancy map. 
We also assume that the grid-based initial trajectory planner in section \ref{subsec: initial trajectory planning} can find a solution when the grid size is $d$.

\subsection{Trajectory Representation}
Due to the differential flatness of quadrotor dynamics, it is known that the trajectory of quadrotor can be represented in a polynomial function with flat outputs in time $t$ \cite{mellinger2011minimum}. However, it is difficult to handle collision avoidance constraints with standard polynomial basis. For this reason, we formulate the trajectory of quadrotors using a piecewise Bernstein polynomial. 
The Bernstein polynomial is the linear combination of Bernstein basis polynomials, and the Bernstein basis polynomial of degree $n$ is defined as follows:
\begin{equation}
    B_{k,n}(t) = {n \choose k}t^{k}(1-t)^{n-k}
\label{eq: bernstein basis}
\end{equation}
for $t\in[0, 1]$ and $k=0,1,...,n$.
The trajectory of the $i^{th}$ quadrotor, $p^{i}(t) \in \mathbb{R}^{3}$, can be represented as $M$-segment piecewise Bernstein polynomials:
\begin{equation}
\begin{alignedat}{2}
    p^{i}(t) = 
    \begin{cases} 
    \ \sum_{k=0}^{n}c^{i}_{1,k}B_{k,n}(\tau_1)   & t \in [T_{0}, T_{1}] \\
    \ \sum_{k=0}^{n}c^{i}_{2,k}B_{k,n}(\tau_2)   & t \in [T_{1}, T_{2}] \\
    \ \vdotswithin{=}                                                 & \vdotswithin{\in} \\    
    \ \sum_{k=0}^{n}c^{i}_{M,k}B_{k,n}(\tau_M) & t \in [T_{M-1}, T_{M}] \\
    \end{cases}
\end{alignedat}
\label{eq: trajectory representation}
\end{equation}
where $\tau_m = \frac{t-T_{m-1}}{T_{m}-T_{m-1}}$, $c^{i}_{m,k}$ is the $k^{th}$ control point of the $m^{th}$ segment of the $i^{th}$ quadrotor's trajectory, and $T_{m-1}, T_{m}$ are the start and end time of the $m^{th}$ segment, respectively. Thus, the decision vector of our optimization problem, $c$, consists of all the control points of $p^{i}(t)$ for $i=1,...,N$.

\subsection{Objective Function}
We define the objective function to minimize the integral of the square of the $\phi^{th}$ derivative:
\begin{equation}
    J = \sum\limits_{i=1}^{N}\int_{T_0}^{T_M}\left\|\frac{d^{\phi}}{dt^{\phi}}p^{i}(t)\right\|^{2}_{2}dt = c^{T}Qc  
\label{eq: objective function}
\end{equation}
where $Q$ is the Hessian matrix of the objective function. In this paper, we set $\phi= 3$ to minimize the jerk of the trajectory, so that the input to the quadrotor becomes less aggressive \cite{mueller2015computationally}.

\subsection{Convex Constraints}

The trajectory must pass the start and goal points and should be continuous up to the $\phi-1^{th}$ derivatives. Also, it must not exceed maximum velocity and acceleration. These constraints can be written in affine equality and inequality constraints respectively:
\begin{equation}
    A_{eq}c = b_{eq}
\label{eq: equality constraints}
\end{equation}
\begin{equation}
    A_{dyn}c \preceq b_{dyn}
\label{eq: dynamic feasible constraints}
\end{equation}

\subsection{Non-Convex Collision Avoidance Constraints}
\begin{figure}[t]
\centering
\includegraphics[width = 0.8\linewidth]{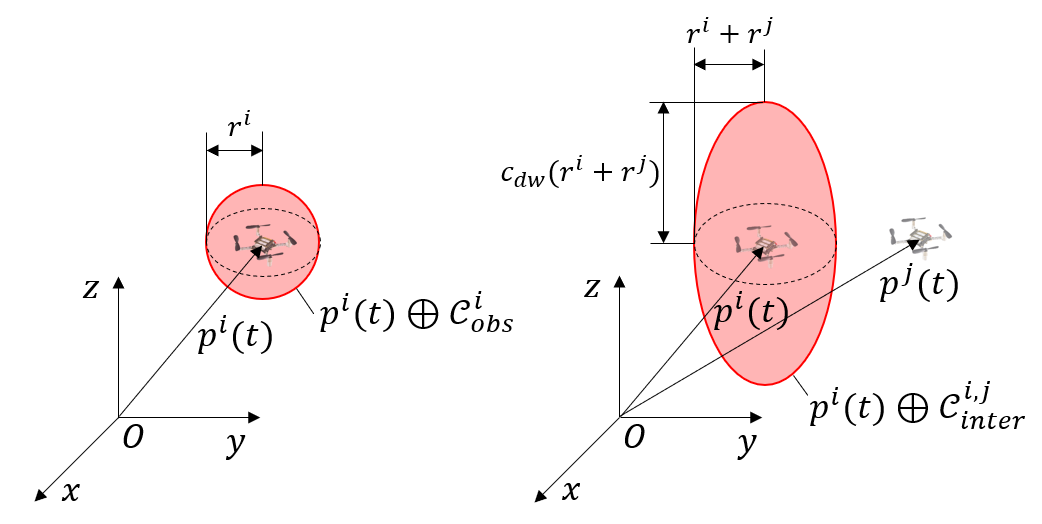}
\caption{
(Left) Obstacle collision model, (Right) Inter-collision model.
}
\label{fig: collision model}
\vspace{-5mm}
\end{figure}

\subsubsection{Obstacle Avoidance Constraints}
We define an obstacle collision model of the $i^{th}$ quadrotor, which models a collision region between a quadrotor and obstacles (See Fig. \ref{fig: collision model}):
\begin{equation}
    \mathcal{C}^i_{obs} = \{ p \in \mathbb{R}^{3} \mid \|p\|^{2}_{2} \leq (r^i)^{2} \}
\label{eq: obstacle collision model}
\end{equation}
The $i^{th}$ quadrotor must satisfy the condition below not to collide with obstacles:
\begin{equation}
    p^{i}(t) \oplus \mathcal{C}^i_{obs} \subset \mathcal{F}, \:\:\:\:\: t \in [T_0, T_M]
\label{eq: obstacle collision avoidance constraint}
\end{equation}
where $\oplus$ is the Minkowski sum.

\subsubsection{Inter-Collision Avoidance Constraints}
% Similar to the obstacle collision model, we define an inter-collision model $\mathcal{C}^{i,j}_{inter}$ an ellipsoidal form, which models a collision region between $i^{th}$ and $j^{th}$ agents
A collision region between $i^{th}$ and $j^{th}$ agents can be expressed with an inter-collision model $\mathcal{C}^{i,j}_{inter}$:
\begin{equation}
    \mathcal{C}^{i,j}_{inter} = \{ p \in \mathbb{R}^{3} \mid p^{T}Ep \leq (r^{i}+r^{j})^{2} \}
\label{eq: inter-collision model}
\end{equation}
where $E$ is $diag([1, 1, 1/(c_{dw})^2])$, and $c_{dw}$ is a coefficient to consider a downwash effect.
The $i^{th}$ agent does not collide with the $j^{th}$ agent if the relative trajectory of the $j^{th}$ agent respect to the $i^{th}$ agent, $p^{i,j}(t)=p^{j}(t)-p^{i}(t)$, satisfies the following condition:
\begin{equation}
    p^{i,j}(t) \cap \mathcal{C}^{i,j}_{inter} = \emptyset, \:\:\:\:\: t \in [T_0, T_M]
\label{eq: inter-collision avoidance constraint}
\end{equation}
Non-convexity of (\ref{eq: obstacle collision avoidance constraint}) and (\ref{eq: inter-collision avoidance constraint}) makes it difficult to directly employ them.
In the next section, we will show the method that relaxes those non-convex constraints to convex ones using relative Bernstein polynomial.

\section{Relative Bernstein Polynomial}
\label{sec: relative bernstein polynomial}
One of the useful properties of the Bernstein polynomial is a convex hull property that the Bernstein polynomial is confined within the convex hull of its control points \cite{zettler1998robustness}.
This property has been used to confine the trajectory to a convex set called safe flight corridor (SFC) for obstacle avoidance \cite{tang2016safe,gao2018online}.
Here, we introduce the method to confine the relative polynomial trajectory to inter-collision-free region by utilizing the convex hull property.

Let the $m^{th}$ segment of $p^{i}(t), p^{j}(t)$ be $p^{i}_m(t), p^{j}_m(t)$ respectively, and $p^{i,j}_{m}(t) = p^{j}_{m}(t)-p^{i}_{m}(t)$ is their relative trajectory. Then  $p^{i,j}_{m}(t)$ can be written as
\begin{equation}
\begin{split}
p^{i,j}_{m}(t) & = \sum\limits_{k=0}^{n}(c^j_{m,k}-c^i_{m,k})B_{k,n}(\tau_{m})\\
               & = \sum\limits_{k=0}^{n}c^{i,j}_{m,k}B_{k,n}(\tau_{m})
\end{split}
\label{eq: relative bernstein polynomial}
\end{equation}
where $c^{i,j}_{m,k} = c^j_{m,k}-c^i_{m,k}$ is the control point of $p^{i,j}_{m}(t)$, which implies that the relative Bernstein polynomial is also Bernstein polynomial.
Thus, by the convex hull property, we can enforce the $i^{th}$ and $j^{th}$ quadrotors not to collide with each other by limiting all control points $c^{i,j}_{m,k}$ within a convex, inter-collision free region.
We call this region a relative safe flight corridor (RSFC).
In this way, we can generate the safe trajectory by adjusting SFC, RSFC. %in our problem.

\section{Definition}
\label{sec: definition}
In the previous work \cite{1909.02896}, we determined RSFC by choosing a proper one among pre-defined RSFC candidates. RSFC candidates were designed to be able to utilize the differential flatness of quadrotor, and so as to achieve fast planning speed. However, it may fail to find a trajectory because a feasible region that satisfies both RSFC and SFC constraints may not exist.
To guarantee the existence of such feasible region as Fig. \ref{fig: collision avoidance constraints}, we precisely define three key terms in this paper: initial trajectory, SFC, and RSFC.

\begin{figure}[t]
\centering
\includegraphics[width = 0.5\linewidth]{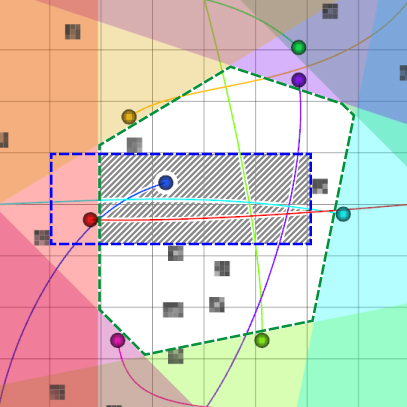}
\caption{
Collision avoidance constraints. The region surrounded by the blue dashed line is safe flight corridor (SFC) for the blue agent, and the region surrounded by green dashed line is the intersection of relative safe flight corridor (RSFC) for the blue agent. To generate a safe trajectory, there must exist intersection between SFC and RSFC (gray shaded area).
}
\label{fig: collision avoidance constraints}
\vspace{-5mm}
\end{figure}

\subsection{Initial Trajectory}
An initial trajectory of the $i^{th}$ quadrotor, $\pi^{i} = \{\pi^{i}_{0},...,\pi^{i}_{M}\}$, is defined as a path that satisfies the following conditions for all $m = 0,...,M$ and $i \neq j$:
\begin{equation}
     \pi^i_0 = s^i, \pi^i_M = g^i
\label{eq: initial trajectory condition0}
\end{equation}
\begin{equation}
     \langle \pi^i_{m-1}, \pi^i_{m} \rangle \oplus \mathcal{C}^i_{obs} \subset \mathcal{F}
\label{eq: initial trajectory condition1}
\end{equation}
\begin{equation}
     \langle \pi^{i,j}_{m-1}, \pi^{i,j}_{m} \rangle \cap \mathcal{C}^{i,j}_{inter} = \emptyset
\label{eq: initial trajectory condition2}
\end{equation}
where $\langle \pi^i_{m-1}, \pi^i_{m} \rangle = \{\alpha\pi^i_{m-1}+(1-\alpha)\pi^i_{m} \mid 0 \leq \alpha \leq 1\}$ is a line segment between waypoints $\pi^i_{m-1}$ and $\pi^i_{m}$, and $\pi^{i,j}_m = \pi^j_{m}-\pi^i_{m}$. \eqref{eq: initial trajectory condition1} shows that the initial trajectory does not collide with obstacles, and \eqref{eq: initial trajectory condition2} means that the agents do not collide with other agents when all the agents move along their initial trajectory at constant velocity.

\subsection{Safe Flight Corridor}
The $m^{th}$ safe flight corridor (SFC) of the $i^{th}$ quadrotor, $\mathcal{S}^{i}_{m}$, is defined as a convex set satisfies following conditions:
\begin{equation}
     \mathcal{S}^{i}_{m} \oplus \mathcal{C}^{i}_{obs} \subset \mathcal{F}
\label{eq: sfc condition1}
\end{equation}
\begin{equation}
     \langle \pi^i_{m-1}, \pi^i_{m} \rangle \subset \mathcal{S}^{i}_{m}
\label{eq: sfc condition2}
\end{equation}
The condition (\ref{eq: sfc condition1}) shows that an agent in SFC does not collide with obstacles, so SFC can be used for obstacle collision avoidance. 
% We note that there always exists a convex set that satisfies eq. (\ref{eq: sfc condition1}) and (\ref{eq: sfc condition2}), for given initial trajectory $\pi^i$ (e.g. $\langle \pi^i_{m-1}, \pi^i_{m} \rangle$).

\subsection{Relative Safe Flight Corridor}
The $m^{th}$ relative safe flight corridor (RSFC) between $i^{th}$ and the $j^{th}$ quadrotor, $\mathcal{R}^{i,j}_{m}$, is defined as a convex set that satisfies the following conditions:
\begin{equation}
     \mathcal{R}^{i,j}_{m} \cap \mathcal{C}^{i,j}_{inter} = \emptyset
\label{eq: rsfc condition1}
\end{equation}
\begin{equation}
     \langle \pi^{i,j}_{m-1}, \pi^{i,j}_{m} \rangle \subset \mathcal{R}^{i,j}_{m}
\label{eq: rsfc condition2}
\end{equation}
If $\mathcal{R}^{i,j}_{m}$ includes $p^{i,j}(t)$ for $t \in [T_{m-1}, T_m]$, then there is no collision between the $i^{th}$ and $j^{th}$ agents for $t \in [T_{m-1}, T_m]$ due to  (\ref{eq: inter-collision avoidance constraint}) and (\ref{eq: rsfc condition1}). For this reason, we can use RSFC to avoid collision between agents. 
% We also note that there always exists a convex set that satisfies above conditions, for given initial trajectories $\pi^i$, $\pi^j$ (e.g. $\langle \pi^{i,j}_{m-1}, \pi^{i,j}_{m} \rangle$).

\section{Method}
In this section, we introduce the efficient trajectory planning algorithm using convex safe corridors.
Alg. \ref{alg: trajectory planning algorithm} shows the process of trajectory planning. 
We first plan initial trajectories (line 1), and we use them to determine safe flight corridor (SFC) (line 3) and relative safe flight corridor (RSFC) (line 5). After that, we compose quadratic programming (QP) problem using initial trajectories and safe corridors (line 8). Finally, we scale the total flight time to satisfy dynamic feasibility constraints (line 9). The detail of each process is described in the following subsections. 

\label{sec: method}
\setlength{\textfloatsep}{5pt}
\begin{algorithm}
\SetAlgoLined
\KwIn{start point $s^i$, goal point $g^i$ for agents $i \in \{1,...,N\}$, 3D occupancy map $\mathcal{E}$}
\KwOut{total flight time $T$,
       trajectory $p^i(t)$ for agents $i \in \{1,...,N\}$, $t \in [0,T]$}
  $\pi = (\pi^{1},...,\pi^{N}) \gets$ planInitialTraj($s^{\forall i}$, $g^{\forall i}, \mathcal{E}$)\;
  \For{$i \gets 1$ \KwTo $N$}{
    $\mathcal{S}^{i} = (\mathcal{S}^{i}_{0},...,\mathcal{S}^{i}_{M}) \gets$ buildSFC($\pi^{i}, \mathcal{E}$)\;
    \For{$j \gets i+1$ \KwTo $N$}{
      $\mathcal{R}^{i,j} = (\mathcal{R}^{i,j}_{0},...,\mathcal{R}^{i,j}_{M}) \gets$ buildRSFC($\pi^{i},\pi^{j}$)\;
    }
  }
  $p^{0}(t),...,p^{N}(t) \gets$ trajOpt($\pi, \mathcal{S}^{\forall i}, \mathcal{R}^{\forall i, j > i})$\;
  $T, p^{0}(t),...,p^{N}(t) \gets$ timeScale($p^{0}(t),...,p^{N}(t)$)\;
  \KwRet{$T, p^{0}(t),...,p^{N}(t)$}
\caption{Trajectory Planning Algorithm}
\label{alg: trajectory planning algorithm}
\end{algorithm}

\subsection{Initial Trajectory Planning}
\label{subsec: initial trajectory planning}
To plan an initial trajectory, we use a graph-based multi-agent pathfinding (MAPF) algorithm. 
Among various MAPF algorithms such as \cite{wagner2011m, sharon2015conflict}, we choose enhanced conflict-based search (ECBS) for the following two reasons: 
(i) ECBS can find a suboptimal solution in a short time. Because the optimal MAPF algorithm is NP-complete \cite{yu2013structure}, it could be better to use a suboptimal MAPF solver with respect to computation time.
(ii) The ECBS algorithm is complete. 
To guarantee completeness of Alg. \ref{alg: trajectory planning algorithm}, individual submodules in the algorithm must be complete.

To utilize the graph-based ECBS in our problem, we formulate the planInitialTrajectory function in line 1, Alg \ref{alg: trajectory planning algorithm} as follows.
First, we translate the given 3D occupancy map into a 3D grid map with grid size $d$. 
Next, we set constraints which determine conflict in the ECBS algorithm to satisfy the condition (\ref{eq: initial trajectory condition2}). 
After that, we give start and goal points as the input and compute the initial trajectory. 
If start and goal points are not located on the 3D grid map, then we use the nearest grid points instead and append the start/goal points to both ends respectively. 

\subsection{Safe Flight Corridor Construction}
\label{subsec: safe flight corridor construction}
Alg. \ref{alg: sfc construction} shows the construction process of safe flight corridor (SFC). 
We initialize SFC to $\langle \pi^{i}_{m-1}, \pi^{i}_{m} \rangle$ to fulfill the condition (\ref{eq: sfc condition2}) (line 3). 
For all direction, we check whether SFC is expandable (line 5-9), and we expand SFC by a pre-specified length  (line 10).
This algorithm guarantees to return convex sets that satisfy the definition of SFC.

\begin{algorithm}
\SetAlgoLined
\KwIn{initial trajectory $\pi^i$, 3D occupancy map $\mathcal{E}$}
\KwOut{safe flight corridor $\mathcal{S}^{i} = (\mathcal{S}^{i}_{0},...,\mathcal{S}^{i}_{M})$}
  $D \gets \{\pm x, \pm y, \pm z\}$\;
  \For{$m \gets 1$ \KwTo $M$}{
      $\mathcal{S}^{i}_{m} \gets \langle \pi^{i}_{m-1}, \pi^{i}_{m} \rangle$\;
      \While{$D$ is not empty}{
        \For{$\mu$ in $D$}{
          \If{$\mathcal{S}^{i}_{m}$ cannot expand to direction $\mu$}{
            $D \gets D \backslash \mu$    
          }
        }
        expand $\mathcal{S}^{i}_{m}$ to all direction in $D$\;
      }
  }
\caption{buildSFC}
\label{alg: sfc construction}
\end{algorithm}
\vspace{-5mm}

\subsection{Relative Safe Flight Corridor Construction}
\begin{figure}
    % \vspace{1.5mm}
    \centering
    \begin{subfigure}[t]{0.21\textwidth}
        \includegraphics[width=\textwidth]{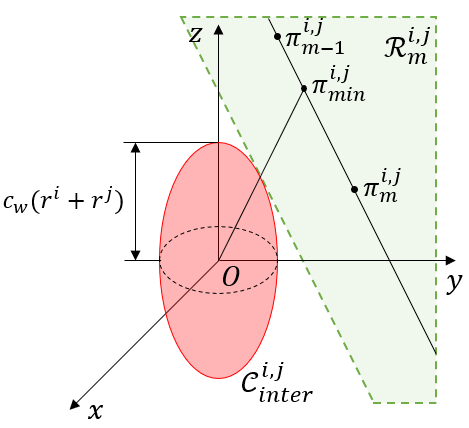}
        \caption{Before coordinate transformation.}
        \label{fig: relative safe flight corridor construction 1}
    \end{subfigure}
    ~ %add desired spacing between images, e. g. ~, \quad, \qquad, \hfill etc. 
    \begin{subfigure}[t]{0.21\textwidth}
        \includegraphics[width=\textwidth]{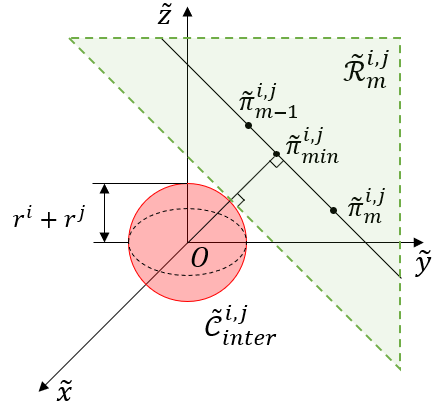}
        \caption{After coordinate transformation.}
        \label{fig: relative safe flight corridor construction 2}
    \end{subfigure}
    \caption{Construction of relative safe flight corridor. The red ellipsoid is an inter-collision model between the quadrotors $i,j$, and the green-shaded region is the relative safe flight corridor (RSFC).
    }
    \label{fig: relative safe flight corridor construction}
\end{figure}
To build RSFC, we first perform affine coordinate transformation $\widetilde{x}=E^{\frac{1}{2}}x$, where $E^{\frac{1}{2}}$ is $diag([1, 1, 1/c_{dw}])$. Then, the inter-collision model $\mathcal{C}^{i,j}_{inter}$ and initial trajectory $\pi^{i,j}$ are transformed to $\widetilde{\mathcal{C}}^{i,j}_{inter}$ and $\widetilde{\pi}^{i,j}$ as shown in Fig. \ref{fig: relative safe flight corridor construction 2}. 
Let $\widetilde{\pi}^{i,j}_{min}$ be the nearest point of $\langle \widetilde{\pi}^{i,j}_{m-1}, \widetilde{\pi}^{i,j}_{m} \rangle$ to the origin. We construct RSFC as follows:
\begin{equation}
     \mathcal{R}^{i,j}_{m} = \{x = E^{-\frac{1}{2}}\widetilde{x} \mid \widetilde{x} \cdot \widetilde{n}_{min} - (r^i+r^j) > 0 \}
\label{eq: rsfc construction}
\end{equation}
where $\widetilde{n}_{min} = \widetilde{\pi}^{i,j}_{min} / \|\widetilde{\pi}^{i,j}_{min}\|$.
As depicted in Fig. \ref{fig: relative safe flight corridor construction 1}, our RSFC is a half-space divided by the plane, which is tangent to the inter-collision model at the $\pi^{i,j}_{min}= E^{-\frac{1}{2}}\widetilde{\pi}^{i,j}_{min}$.
We note that the convex set in (\ref{eq: rsfc construction}) satisfies the definition of RSFC, but we omit the proof due to page limits.  %(in the longer archive version  ??)

\subsection{Sequential Trajectory Optimization}
\label{subsec: trajectory optimization}
Optimizing all control points of polynomials at once can cause the scalability problem because the time complexity of the QP solver is $O(n^3)$.
Here, we propose an efficient sequential optimization method using \textit{dummy agents}. 
% This method can improve scalability by dividing the big problem into several small ones, and prevent deadlock by utilizing dummy agents.

\begin{algorithm}
\SetAlgoLined
\KwIn{initial trajectory $\pi$, safe flight corridor $\mathcal{S}^{\forall i}$, relative safe flight corridor $\mathcal{R}^{\forall i,j>i}$}
\KwOut{trajectory $p^i(t)$ for agents $i \in b$, $t \in [0,T]$}
  $p_{dmy}(t) = (p^{0}_{dmy}(t),...,p^{N}_{dmy}(t)) \gets$ planDummy($\pi$)\;
  \For{$l \gets 1$ \KwTo $N_{b}$}{
    $b \gets $ agents in $l^{th}$ batch\;
    $p^{b}(t) \gets$ solveQP($\pi^{b}, \mathcal{S}^{b}, \mathcal{R}^{\forall i, j > i}, p^{\forall i \notin b}_{dmy}(t)$)\;
    $p_{dmy}(t) \gets p(t)$\;
  }
  \KwRet{$p^{0}(t),...,p^{N}(t)$}
\caption{trajOpt}
\label{alg: trajopt}
\end{algorithm}
\vspace{-2mm}

Alg. \ref{alg: trajopt} shows the process of sequential optimization. First, we generate trajectories for dummy agents $p_{dmy}(t)$ using the following control points (line 1):
\begin{equation}
    c^{i}_{m,k} = 
    \begin{cases} 
    \ \pi^{i}_{m-1}   & k = 0,...,\phi-1 \\
    \ \pi^{i}_{m}     & k = n-(\phi-1),...,n \\
    \ x \in \langle \pi^i_{m-1}, \pi^i_{m} \rangle & else
    \end{cases}
\label{eq: feasible constraints}
\end{equation}
% These dummy agents are used to prevent the previously planned trajectory from blocking the space for the other agents.
% Next, we divide the agents into $N_{b}$ batches, and we optimize the trajectory by solving the below QP problem for the batch $b$ (line 3-4, Fig. \ref{fig: sequential planning 2}):
Next, we divide the agents into $N_{b}$ batches, and we solve the below QP problem for the batch $b$ (line 3-4):
\begin{equation}
\begin{aligned}
& \text{minimize}     & & c^{T}Qc \\
& \text{subject to}   & & A_{eq}c = b_{eq} \\
&                     & & c^i_{m,k} = \text{control points of } p^{i}_{dmy}(t), & & \forall i \notin b, m, k \\
&                     & & c^i_{m,k} \in \mathcal{S}^{i}_{m},             & & \forall i \in b, m, k\\
&                     & & c^j_{m,k}-c^i_{m,k} \in \mathcal{R}^{i,j}_{m}  & & \forall i, j>i, m, k\\
\end{aligned}
\label{eq: trajectory optimization}
\end{equation}
where $c \in \mathbb{R}^{\frac{N}{N_{b}}M(n+1)}$, $A_{eq} \in \mathbb{R}^{\frac{N}{N_{b}}(M+1)\phi \times \frac{N}{N_{b}}M(n+1)}$, and the number of inequality constraints is $(N-\frac{1}{2}(\frac{N}{N_{b}}+1))\frac{N}{N_{b}}M(n+1)$. $p^{i}_{dmy}(t)$ is the trajectory for $i^{th}$ dummy agent.
% At last, we replace the trajectory of dummy agents to the previously planned one (line 5, Fig. \ref{fig: sequential planning 3}), and plan the trajectory for the next batch sequentially (Fig. \ref{fig: sequential planning 4}).
At last, we replace the trajectory of dummy agents to the previously planned one (line 5), and plan the trajectory for the next batch sequentially.

Fig. \ref{fig: sequential planning} visualizes the overall process. 
For each iteration, we deploy dummy agents except the agents in the current batch (Fig. \ref{fig: sequential planning 1}).
Then, we plan the trajectory for the current batch to avoid dummy agents (Fig. \ref{fig: sequential planning 2}).  
After that, agents in the current batch are used as dummy agents at the next iteration (Fig. \ref{fig: sequential planning 3}). At the end of the iteration, collision-free trajectories are found without deadlock because all the agents are planned to avoid the previous batch (Fig. \ref{fig: sequential planning 4}).

This sequential method can achieve better scalability because we can avoid the high time complexity of QP solver by increasing the number of the batch as the number of agents increases while keeping the same number of decision variables of QP. 
Furthermore, we can prove that \eqref{eq: trajectory optimization} consists of feasible constraints, which means that our method does not cause optimization failure due to infeasible constraints.
\begin{theorem}
If $n \geq 2\phi-1$, then there exists a decision vector $c$ that satisfies the constraints of eq. (\ref{eq: trajectory optimization}) for all iterations.
\label{theorem: feasible constraints}
\end{theorem}

\begin{proof}
For any iteration $l$, let us assign the decision vector $c$ as \eqref{eq: feasible constraints} for $i \in b$ and $m = 1,...,M$.
Then $c$ satisfies the waypoint constraints due to \eqref{eq: initial trajectory condition0}. 
$p^i(t)$ is continuous up to $\phi-1^{th}$ derivatives at $t=T_{m}$ for $m = 1,...,M-1$ because $c^{i}_{m,n-(\phi-1):n} = c^{i}_{m+1,0:\phi-1} = \pi^{i}_{m}$. $c$ also fulfills safe corridor constraints due to \eqref{eq: sfc condition2} and \eqref{eq: rsfc condition2}. Note that trajectories in the iteration $1,...,l-1$ do not collide with trajectories in the current batch because they are planned to avoid dummy agents which consist of control points \eqref{eq: feasible constraints}.
Thus, $c$ is the decision vector that satisfies the constraints of \eqref{eq: trajectory optimization}.
\end{proof}

% \section{Sequential Planning for Scalability}
% \label{sec: sequential planning for scalability}
\begin{figure}
    % \vspace{1.5mm}
    \centering
    \begin{subfigure}[t]{0.19\textwidth}
        \includegraphics[width=\textwidth]{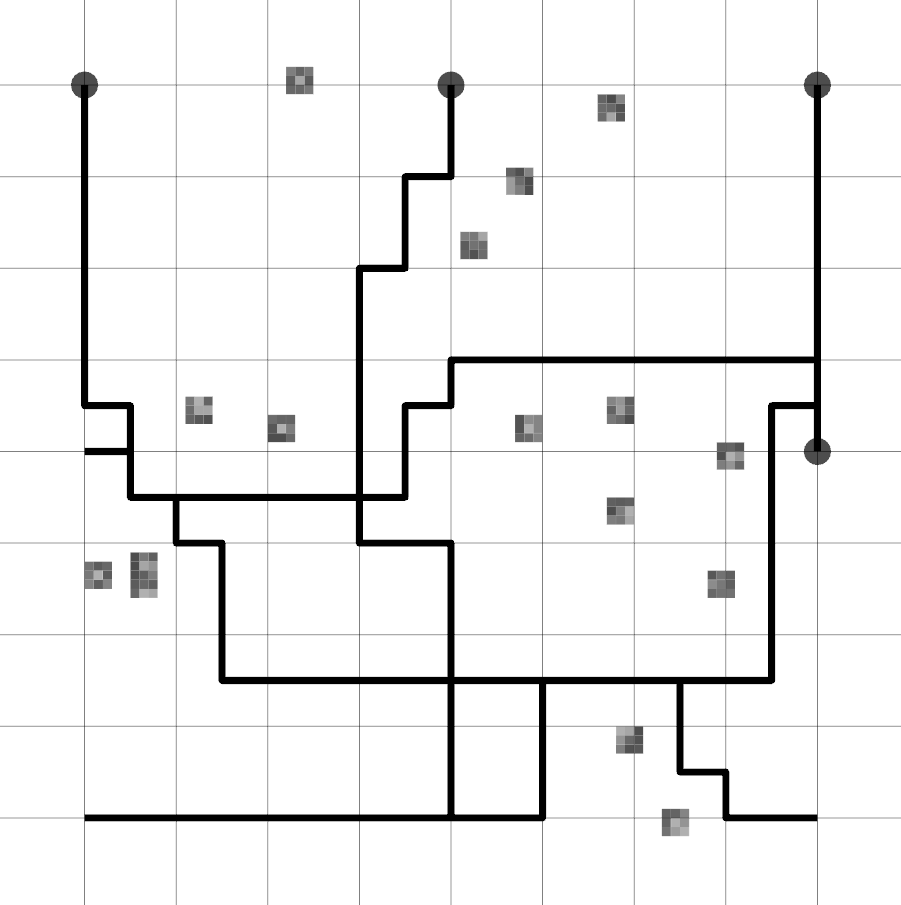}
        \caption{Deploy dummy agents.}
        \label{fig: sequential planning 1}
    \end{subfigure}
    ~ %add desired spacing between images, e. g. ~, \quad, \qquad, \hfill etc. 
    \begin{subfigure}[t]{0.19\textwidth}
        \includegraphics[width=\textwidth]{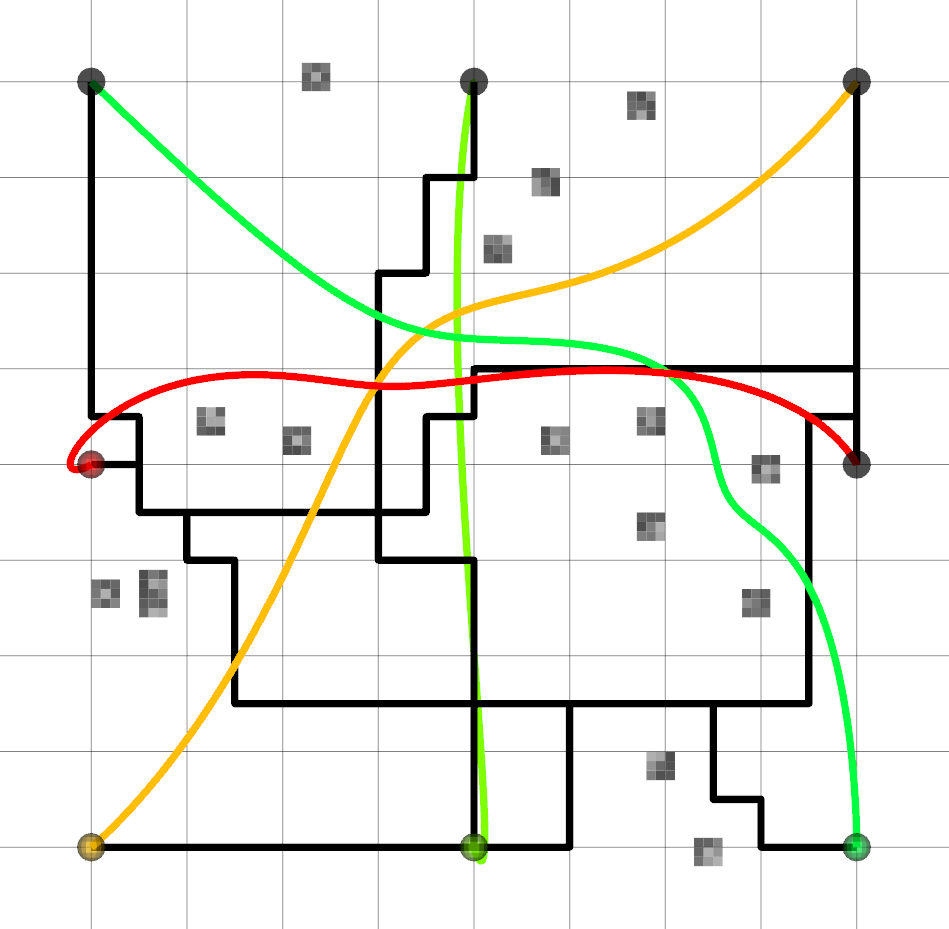}
        \caption{Plan for one batch. 
        %Plan trajectories for one batch.
        }
        \label{fig: sequential planning 2}
    \end{subfigure}
    ~ %add desired spacing between images, e. g. ~, \quad, \qquad, \hfill etc. 
    \begin{subfigure}[t]{0.19\textwidth}
        \includegraphics[width=\textwidth]{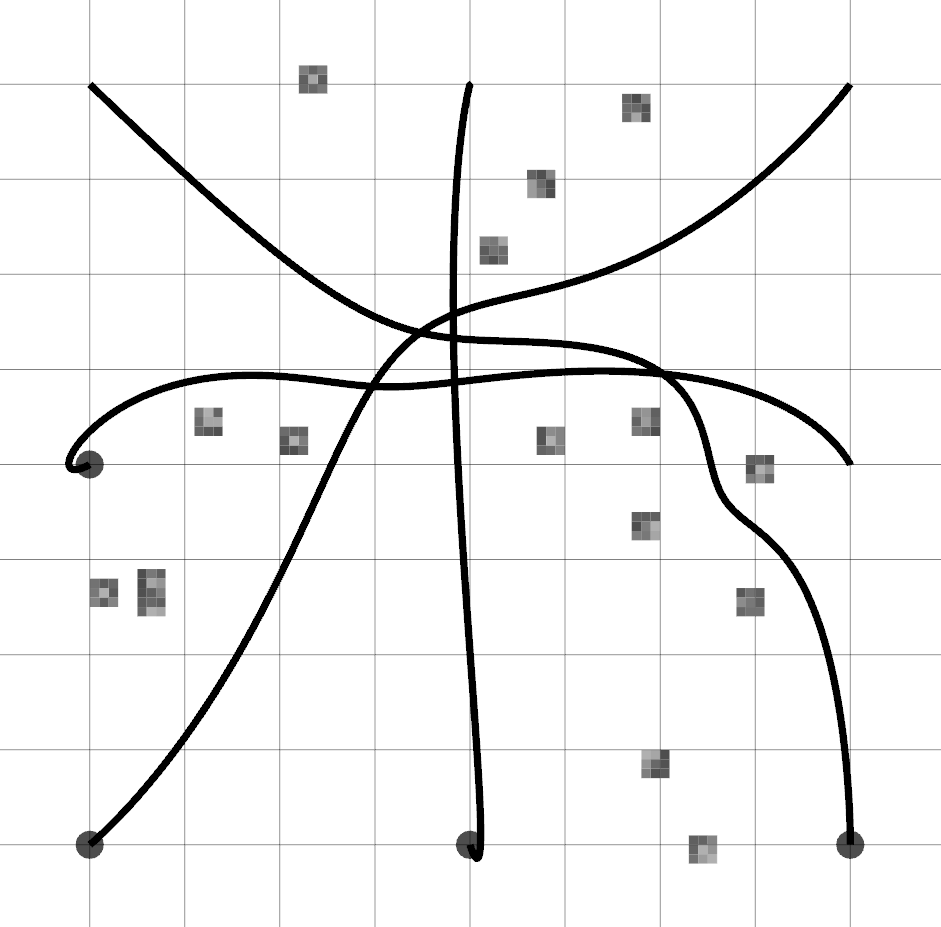}
        \caption{Replace dummy agents with the previous batch.}
        \label{fig: sequential planning 3}
    \end{subfigure}
    ~ %add desired spacing between images, e. g. ~, \quad, \qquad, \hfill etc. 
    \begin{subfigure}[t]{0.19\textwidth}
        \includegraphics[width=\textwidth]{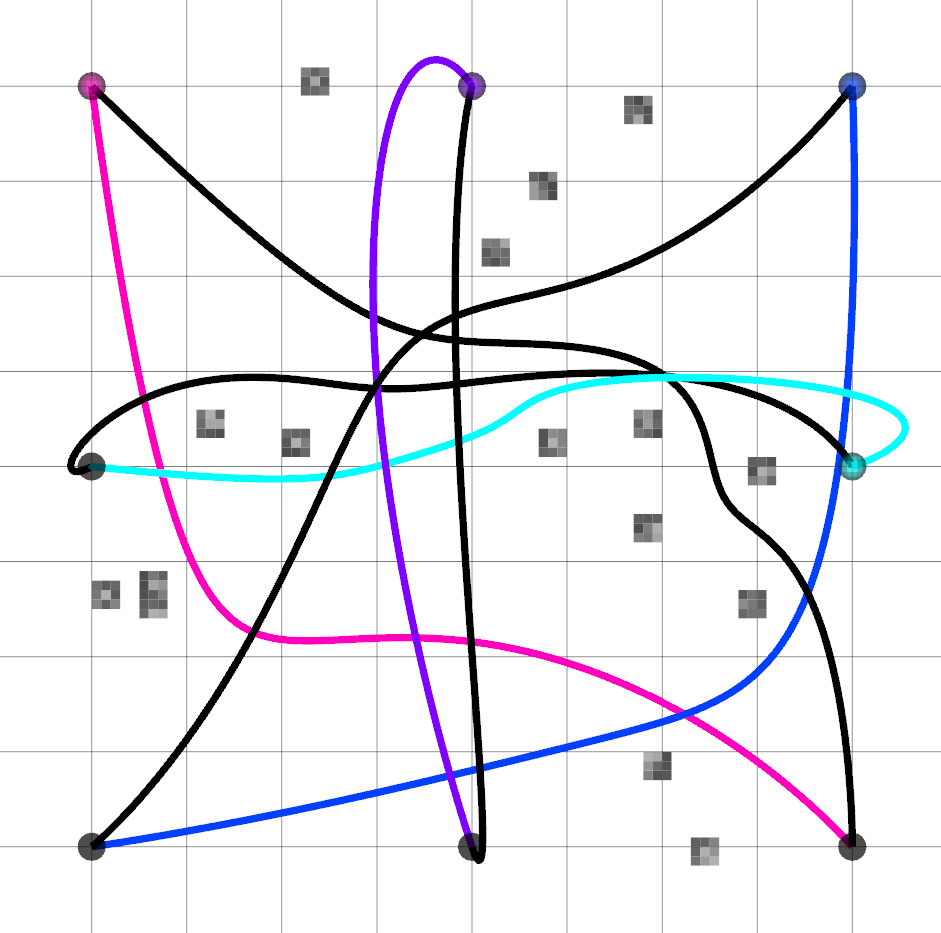}
        \caption{Plan for next batch.}
        \label{fig: sequential planning 4}
    \end{subfigure}
    \caption{Sequential planning with dummy agents when $N_{b}=2$. Dummy agent is depicted as a black circle, and agent in the current batch is depicted as a colored circle. For each iteration, we plan a trajectory for the current batch (color line) that avoids the trajectory of dummy agents (black line).
    %, and the previously planned trajectory is used as a dummy agent for the next iteration.
    }
    \label{fig: sequential planning}
\end{figure}

In (\ref{eq: trajectory optimization}), we do not consider dynamic limits in the QP problem because they can be infeasible constraints for QP. Instead, similar to \cite{honig2018trajectory}, we scale the total flight time for all agents uniformly after optimization (line 9 of Alg. \ref{alg: trajectory planning algorithm}).

\section{EXPERIMENTS}
\label{sec: experiments}
\subsection{Implementation Details}
The proposed algorithm is run in C++ and executed the proposed algorithm on a PC running Ubuntu 18.04. with Intel Core i7-7700 @ 3.60GHz CPU and 16G RAM.
We model the quadrotor with radius $r^{\forall i} = 0.15$m, maximum velocity $v^{\forall i}_{max} = 1.7m/s$, maximum acceleration $a^{\forall i}_{max} = 6.2m/s^2$ and downwash coefficient $c_{dw} = 2$ based on the specification of Crazyflie 2.0 in \cite{debord2018trajectory}.
We use the Octomap library \cite{hornung2013octomap} to represent the 3D occupancy map and use CPLEX QP solver \cite{cplex201612} for trajectory optimization.
% the dynamicEDT3D library \cite{lau2013efficient} to convert occupancy map to 3D grid map.
The degree of polynomials is determined to $n = 5$ to satisfy the assumption in the Theorem \ref{theorem: feasible constraints}. %\textcolor{green}{($n \geq 2\phi-1$).}
We plan the initial trajectory in 3D grid map with grid size $d=0.5$ m, and set suboptimal bound of ECBS to $1.3$.

\subsection{Comparison with Previous Work}
To validate the performance of our proposed algorithm, we compared the result with previous work \cite{1909.02896}. 
% \cite{1909.02896} also utilizes RSFC for inter-collision avoidance, but it uses the RSFC candidate method, which constructs RSFC by choosing the proper one of pre-defined RSFC candidates.
We conducted the simulations in 50 random forests. Each forest has a size of 10 m $\times$ 10 m $\times$ 2.5 m and contains randomly deployed 20 trees of size 0.3 m $\times$ 0.3 m $\times$ 1--2.5 m.
We assigned start point of quadrotors in a boundary of the xy-plane in 1 m height, and the goal points at the opposite to their start position as shown in Fig. \ref{fig:simulation}.
% Fig. \ref{fig:simulation} shows the simulation result with 16 agents.

\begin{figure}[t]
\centering
\includegraphics[width = 0.5\textwidth]{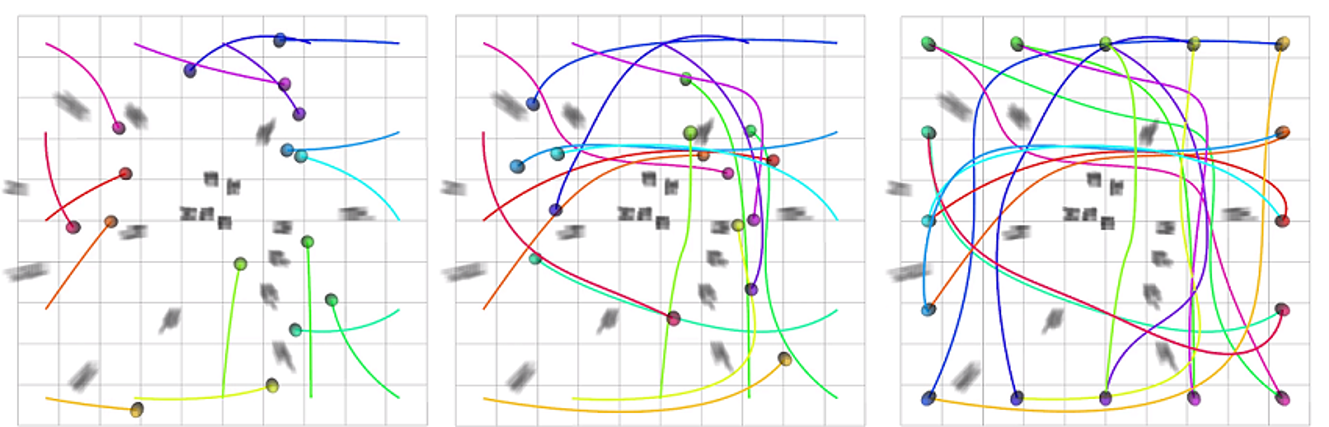}
\caption{Trajectory generated for 16 agents in a 10 m $\times$ 10 m $\times$ 2.5 m random forest. Agents are marked with colored circles at the goal (assigned to the opposite of the start points), along with their trajectories.
}
\label{fig:simulation}
\end{figure}

\subsubsection{Success Rate}
We executed the simulation with 16 agents, and measured the success rate by the size of agents.
As shown in the left graph of Fig. \ref{fig: prev vs rbp}, both methods show a $100\%$ success rate in 50 random forest when the radius of agents is small, but the success rate of \cite{1909.02896} decreases as the size of agents increases.
It is because the larger the agent size, the smaller the space for agents can exist, which lead to the higher probability that the constraints for SFC and RSFC are infeasible each other.
On the contrary, the proposed method shows a perfect success rate for all case because we design SFC and RSFC to feasible each other.

\subsubsection{Solution Quality}
As depicted in the right graph of Fig. \ref{fig: prev vs rbp}, the proposed algorithm shows better performance with respect to both objective cost and computation time compare to previous work when the number of the batch $N_b$ is more than one.
It can generate a trajectory for 64 agents in 6.36 s ($N_b=16$), and it has $78\%$ ($N_b=1$), $53\%$ ($N_b=16$) less objective cost.
Note that we can adjust $N_b$ depending on the desired objective cost and computation time.

\begin{table}[t!]
\caption{Computation time comparison with previous work \cite{1909.02896}. The numbers in parentheses represent the computation time increment when the number of agents is doubled.} \vspace{-3mm}
\label{table: scalablilty}
\begin{center}
\begin{tabularx}{\linewidth}{c|*{5}{X}}
\toprule
% \hline 
 & \multicolumn{5}{c}{Computation Time (s)} \tabularnewline
% \hline 
Agents & \centering 4 & \centering 8 & \centering 16 & \centering 32 & \centering 64 \tabularnewline
\hline 
\cite{1909.02896} & \centering 0.093 & \centering 0.19 ($\times 2.0$) &  \centering 0.81 ($\times 4.3$) & \centering 5.30 ($\times 6.5$)& \centering 51.1 ($\times 9.6$) \tabularnewline
% \hline
Proposed ($N_b=1$) & \centering 0.11 & \centering 0.29 ($\times 2.7$)& \centering 1.15 ($\times 3.9$) & \centering 11.1 ($\times 9.6$) & \centering 197.0 ($\times 17.8$) \tabularnewline
% \hline  
Proposed ($N/N_b=4$) & \centering 0.11 & \centering 0.23 ($\times 2.2$) & \centering 0.59 ($\times 2.5$) & \centering 1.55 ($\times 2.6$) & \centering 6.36 ($\times 4.1$) \tabularnewline
% \hline
\hline
\end{tabularx}
\end{center}
\vspace{-5mm}
\end{table}

\begin{figure}
    \centering
    \begin{subfigure}[t]{0.24\textwidth}
        \centering
        \includegraphics[width=\textwidth]{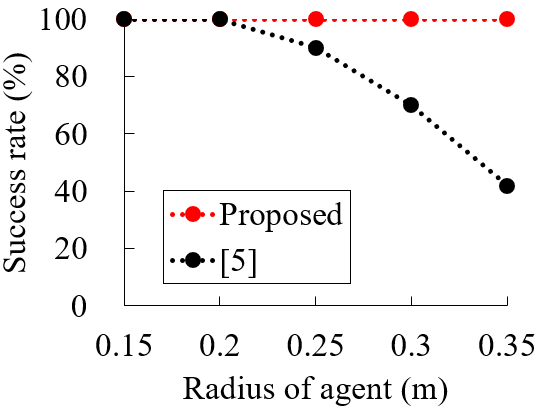}
        % \caption{Success rate of two trajectory generation methods for 16 agents.}
        % \caption{}
        % \label{fig: successrate}
    \end{subfigure}%
    \hfill %add desired spacing between images, e. g. ~, \quad, \qquad, \hfill etc. 
    \begin{subfigure}[t]{0.24\textwidth}
        \centering
        \includegraphics[width=\textwidth]{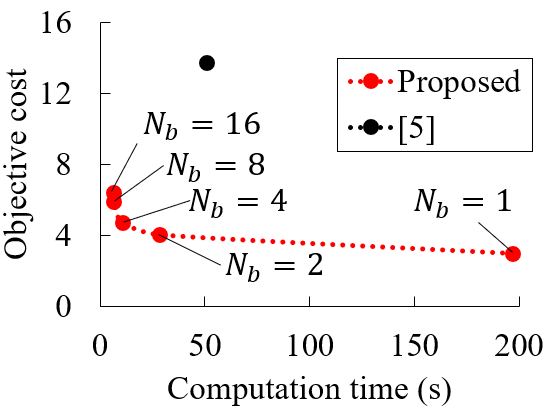}
        % \caption{Cost vs computation time for 64 agents.}
        % \caption{}
        % \label{fig: costvscomp}
    \end{subfigure}
    
    % \centering
    % \includegraphics[width = 1\textwidth]{figure/comparision_prev.jpg}
    
    \caption{
    Comparison with previous work \cite{1909.02896}. (Left) The success rate for 16 agents, (Right) Objective cost and computation time for 64 agents by the number of batches $N_b$.
    }
    \label{fig: prev vs rbp}
\end{figure}

\subsubsection{Scalability Analysis}
The computation time increment by the number of agents is shown in Table \ref{table: scalablilty}. When the number of agents is small, the computation time increases linearly, regardless of the trajectory optimization method, but it follows the time complexity of QP solver as the number of agents increases if we do not adopt the sequential optimization method.
On the other hand, if we maintain the size of the batch ($N/N_b$), it still shows good scalability with the high number of agents. %이유 써줘야할듯?

% \begin{table*}[t!]
% \caption{Performance comparison with previous work \cite{1909.02896}}
% \label{table:scale}
% \begin{center}
% \begin{tabularx}{\linewidth}{c||*{5}{X}|X}
% \toprule
% % \hline 
%  & \multicolumn{5}{c|}{Computation Time (s)} & \centering Objective Cost \tabularnewline
% % \hline 
% Agents & \centering 4 & \centering 8 & \centering 16 & \centering 32 & \centering 64 & \centering 64 \tabularnewline
% \hline 
% \cite{1909.02896} & \centering 0.093 & \centering 0.19 &  \centering 0.81 & \centering 5.30 & \centering 51.1 & \centering 13.7 \tabularnewline
% % \hline
% prop ($N_b=1$) & \centering 0.11 & \centering 0.29 & \centering 1.15 & \centering 11.1 & \centering 197.0 & \centering 2.98 \tabularnewline
% % \hline  
% prop ($N_b=2$) & \centering 0.091 & \centering 0.23 & \centering 0.68 & \centering 2.97 & \centering 28.7 & \centering 4.02 \tabularnewline
% % \hline
% prop ($N_b=4$) & \centering 0.089 & \centering 0.19 & \centering 0.59 & \centering 1.71 & \centering 10.8 & \centering 4.68 \tabularnewline
% % \hline
% prop ($N_b=8$) & \centering - & \centering 0.19 & \centering 0.45 & \centering 1.55 & \centering 6.38 & \centering 5.86 \tabularnewline
% % \hline
% prop ($N_b=16$) & \centering - & \centering - & \centering 0.45 & \centering 1.17 & \centering 6.36 & \centering  6.39 \tabularnewline
% \hline
% \end{tabularx}
% \end{center}
% \end{table*}

\subsection{Comparison with SCP-based Method}
We compared the proposed algorithm with SCP-based method \cite{augugliaro2012generation}.
Experiments are done in 10 m $\times$ 10 m $\times$ 2.5 m empty space with 8 agents. Start position and goal points are same as the previous experiment as shown in Fig. \ref{fig: scp vs rbp}, and we assigned the same total flight time to both algorithm.
% When we run the SCP method, we did not consider the case when the time step $h$ is over 1 because it ignores collision avoidance constraint at all.

Table \ref{table: scp vs rbp} shows that the proposed algorithm requires less computation time for all the cases, and this result does not change when we stop the SCP at the first iteration with collision avoidance constraints. 
The third column shows the safety margin ratio of each method. Safety margin ratio is calculated by $\text{min}_{i,j} d^{i,j}_{min}/(r^i+r^j)$, where $d^{i,j}_{min}$ is a minimum distance between two agents $i,j$. Safety margin ratio must be over $100\%$ to guarantee the collision avoidance, however, SCP-based method does not satisfy this because SCP checks only collision avoidance between discrete points on each trajectory. On the contrary, the proposed method satisfies the safety condition completely.

Although the proposed method perform better in computation time and safety margin, it has longer total flight distance compared to the SCP method. It is because our initial trajectory is not optimal respect to total flight distance in non-grid space. Thus, we need to plan initial trajectory considering total flight distance, and leave it as future work.
% To solve this problem, we need to plan initial trajectory considering total flight distance, and 

\begin{table}[t!]
\caption{Comparison of proposed algorithm and SCP-based method.}
\label{table: scp vs rbp}
\begin{center}
\begin{tabularx}{\linewidth}{c|*{4}{X}}
\toprule
& \centering Comp. Time per Iter. (s) & \centering Total Comp. Time (s) &  \centering Safety Margin Ratio & \centering Total Flight Dist. (m) \tabularnewline
\hline 
SCP ($h=1.0$ s) &  \centering 0.78 & \centering 2.80 &  \centering $12\%$ & \centering \textbf{77.29} \tabularnewline
SCP ($h=0.5$ s) &  \centering 5.5 & \centering 20.5 &  \centering $81\%$ & \centering 77.36 \tabularnewline
% \hline
SCP ($h=0.34$ s) & \centering 16.2 & \centering 60.4 & \centering $92\%$& \centering 77.38 \tabularnewline
% \hline  
SCP ($h=0.25$ s) & \centering 42.1 & \centering 156.6 & \centering $96\%$& \centering 77.40 \tabularnewline
% \hline
Proposed ($N_b=1$) & \centering - & \centering \textbf{0.65} & \centering \textbf{101}$\%$& \centering 90.74  \tabularnewline
\hline
\end{tabularx}
\end{center}
\vspace{-5mm}
\end{table}

\begin{figure}
    % \vspace{1.5mm}
    \centering
    \begin{subfigure}[t]{0.17\textwidth}
        \centering
        \includegraphics[width=\textwidth]{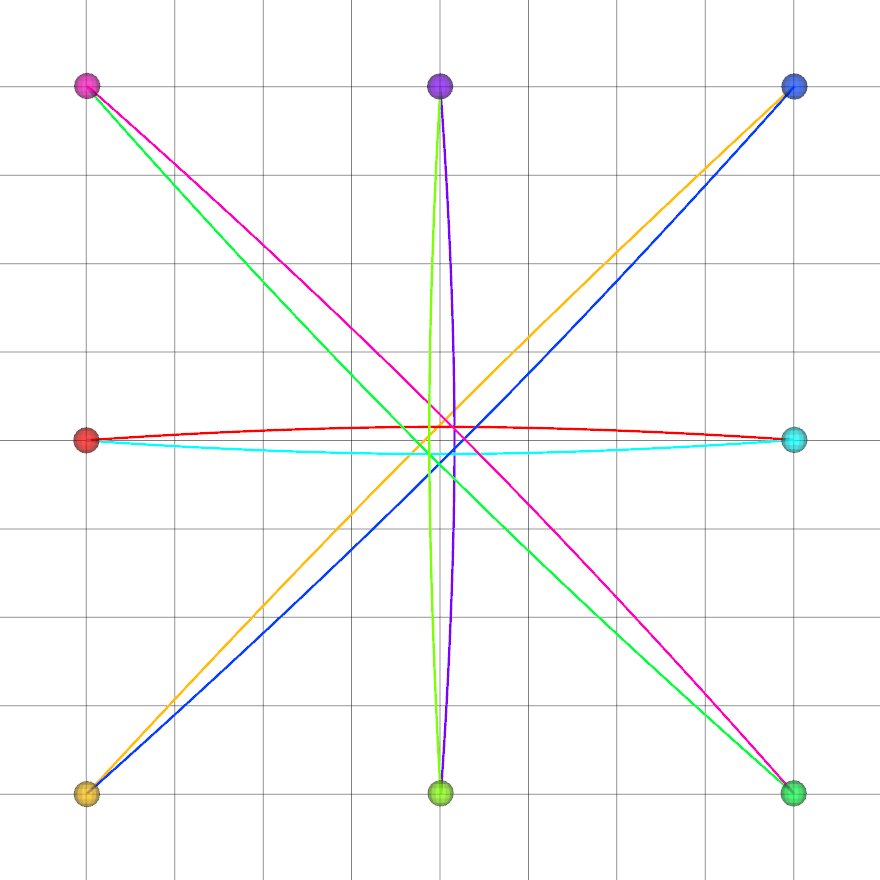}
        \caption{SCP-based \cite{augugliaro2012generation}}
        \label{fig: scp}
    \end{subfigure}%
    ~ %add desired spacing between images, e. g. ~, \quad, \qquad, \hfill etc. 
    \begin{subfigure}[t]{0.17\textwidth}
        \centering
        \includegraphics[width=\textwidth]{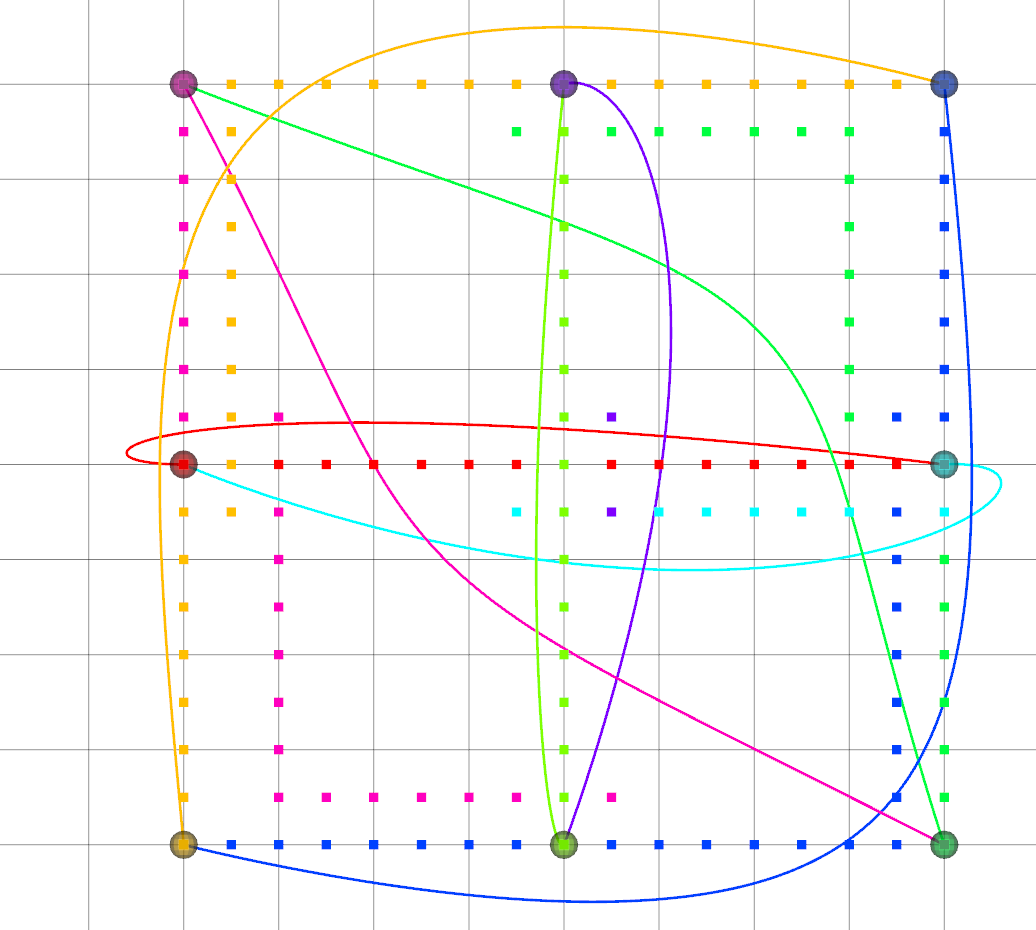}
        \caption{Proposed}
        \label{fig: rbp}
    \end{subfigure}
    \caption{
    Trajectory planning result of the propose algorithm and SCP-based method in empty space. The dots in (b) are the initial trajectory of corresponding agents.
    }
    \label{fig: scp vs rbp}
\end{figure}

% \begin{table*}[t!]
% \caption{Comparison of proposed algorithm and SCP-based method.}
% \label{table: scp vs rbp}
% \begin{center}
% \begin{tabularx}{\linewidth}{c|*{4}{X}}
% \toprule
% & \centering Comp. Time per Iteration (s) & \centering Total Comp. Time (s) &  \centering Min. Dist. btw Agents (m) & \centering Total Flight Distance (m) \tabularnewline
% \hline 
% SCP ($h=1.0$) &  \centering 0.78 & \centering 2.80 &  \centering 0.036 ($12\%$) & \centering 77.29 \tabularnewline
% SCP ($h=0.5$) &  \centering 5.5 & \centering 20.5 &  \centering 0.243 ($81\%$)& \centering 77.36 \tabularnewline
% % \hline
% SCP ($h=0.34$) & \centering 16.2 & \centering 60.4 & \centering 0.275 ($92\%$)& \centering 77.38 \tabularnewline
% % \hline  
% SCP ($h=0.25$) & \centering 42.1 & \centering 156.6 & \centering 0.289 ($96\%$)& \centering 77.40 \tabularnewline
% % \hline
% prop ($N_b=1$) & \centering - & \centering 0.65 & \centering 0.304 ($101\%$)& \centering 90.74  \tabularnewline
% \hline
% \end{tabularx}
% \end{center}
% \end{table*}

\subsection{Flight Test}
We conducted real flight test with 6 Crazyflie 2.0 quadrotors in a 5 m x 7 m x 2.5 m space. We used Crazyswarm \cite{preiss2017crazyswarm} to follow the pre-computed trajectory, and we used Vicon motion capture system to obtain the position information at 100 Hz.
The snapshot of the flight test is shown in Fig. \ref{fig: flight test}, and the full flight is presented in the supplemental video.

\section{CONCLUSIONS}
\label{sec: conclusions}
We presented an efficient trajectory planning algorithm for multiple quadrotors in obstacle environments, combining the advantages of grid-based and optimization-based planning algorithm.
Using relative Bernstein polynomial, we reformulated trajectory generation problem to convex optimization problem, which guarantees to generate continuous, collision-free, and dynamically feasible trajectory. We improved the scalability of the algorithm by using sequential optimization method, and we proved overall process does not cause the failure of optimization if there exist initial trajectory. 
The proposed algorithm shows considerable reduction in computation time and objective cost compared to our previous work, and it shows better performance in computation time and safety, compared to SCP-based method. 
  
In future work, we plan to develop initial trajectory planner that optimizes total flight distance in non-grid space, and we will extend our work to dynamic obstacle environment.

\newpage
\printbibliography

\addtolength{\textheight}{-12cm}   % This command serves to balance the column lengths

\end{document}